\documentclass[conference]{IEEEtran}

\usepackage[square, numbers, comma, sort&compress]{natbib}

\usepackage{amsmath}
\usepackage{graphicx}
\usepackage{amsthm}
\usepackage{caption}
\graphicspath{{figures/}}
\ifCLASSINFOpdf
\else
\fi
\newtheorem{lem}{lemma1}

\begin{document}
%
\title{An Exclusion zone for Massive MIMO With Underlay  D2D Communication }

\author{\IEEEauthorblockN{Salah Eddine Hajri and Mohamad Assaad}
\IEEEauthorblockA{Laboratoire des Signaux et Systemes (L2S, CNRS UMR8506), CentraleSupelec \\
Gif-sur-Yvette, France\\
Email: \{Salaheddine.hajri,\; Mohamad.Assaad\}@centralesupelec.fr}

}


%


\maketitle

\begin{abstract}
Fifth generation networks will incorporate a variety of new features in wireless networks such as data offloading, D2D communication, and Massive MIMO. 
Massive MIMO is specially appealing since it achieves huge gains while enabling simple processing like MRC receivers. It suffers, though, from a major shortcoming refereed to as pilot contamination.
 In this paper we propose a frame-work in which, a D2D underlaid Massive MIMO system is implemented  and we will prove that this scheme can reduce the pilot contamination problem while enabling an optimization of the system spectral efficiency.
The D2D communication will help maintain the network coverage while allowing a better channel estimation to be performed.

\end{abstract}


%
\IEEEpeerreviewmaketitle

\section{Introduction}

The rapid increase in smart-phones number  allowed the emergence of new services that although, enriching for the user experience, lead to an exponential growth in the traffic volume over the wireless networks. Since the network resources especially the  used spectrum, are limited, techniques like Massive MIMO, D2D communication, data offloading and small cell networks were proposed in order to  provide the necessary increase in network capacity for future generation networks.
In addition to traffic volume increase, the issue of energy efficiency of the network has to be dealt with.
To this end two of the most promising techniques are Massive MIMO and D2D underlaid network.

Multiuser Massive MIMO systems were intensively investigated and it was proven that with simple processing this systems allow an increase in network capacity with the possibility of a more energy efficient network. Nevertheless, Massive MIMO systems are no without shortcomings. 
In fact while the large number of antennas mitigate the effect of uncorrelated noise and fast fading \cite{Noncooperative}, the interference caused by pilot reuse persist\cite{Noncooperative}. The pilot contamination problem was addressed several times, in \cite{blind} a subspace projection method was used to mitigate the pilot contamination problem.In \cite{pilots}, the authors proposed a covariance-aided channel estimation framework and a coordinated pilot assignment strategy in order to reduce pilot contamination impact.
Another promising technology is Device to Device communication. This technology allows the user equipment to communicate directly using the cellular network resources instead of going through the cellular infrastructure. 

With enabled D2D communication,  proximity-based services  become possible \cite{overview}. This proximity enable an increase in spectral and energy efficiency. It make also more efficient data offloading  schemes possible. 
Combining D2D communication and Massive MIMO, has been proposed in \cite{Gesbert1} and \cite{Gesbert2} for MIMO networks in FDD mode and has proven to enable more efficient forms of feedback.
In this paper, we investigate the coexistence of the two technologies in TDD mode. In order to reduce the pilot contamination phenomenon, we consider in addition an exclusion zone framework and develop an optimization framework that determines the best exclusion zone size and power allocation.

\section{ Contribution and outcomes}
The contributions of our work are summarized as follows:
\begin{enumerate}
\item A base station controlled mode selection scheme:\\
In this paper we propose a mode selection scheme  based on the user's position. This scheme divide the cell into  two parts where  cellular mode and D2D mode are not allowed to coexist\cite{Guard}. The Base station will fix the radius of the exclusion zone in a scalable way based on the user density, user transmit power, average cellular SINR , and average  interference on D2D links.

\item Impact of the mode selection on Pilot contamination:\\
In this paper we consider a realistic framework in which the Base station uses Pilot-based CSI estimation only for cellular users.
This enable us to assess the impact of our proposed mode selection scheme on a major shortcoming of Massive MIMO systems which is Pilot contamination. We find that the mean square error of channel estimation due to Pilot reuse is reduced in our system but vary depending on whether D2D communication is allowed or not during training phase.
\item Average SINR  Optimization:\\
We develop the analytical expressions of the average SINR for a reference cellular user and the average interference on a reference D2D link. Using this expressions we introduce a quasi-concave optimization problem  in which cellular user power and the exclusion region radius are  defined so that the average cellular SINR is maximized while an acceptable interference level on the D2D communication is guaranteed. By using the average interference and average SINR  we are able to perform optimization  without the need for full CSI of all users. This is quite practical in a dense system like the one considered in this paper.
 \end{enumerate}

\section{ SYSTEM MODEL AND PRELIMINARIES}

In this paper we consider  a Multi-user  multi cell MIMO system with a large excess of base
station antennas operating in  time-division duplex (TDD) mode  in which the   coherence time of  the channel  is divided between a training and a data transmission phase.
The base station will use channel reciprocity and a mutually orthogonal set of pilots in order to estimate the channel  of each scheduled cellular user.

We consider also that in addition to the cellular users the network is  underlaid with D2D terminals using the same  frequency spectrum as the cellular users.
 \begin{figure}[h]
	\hspace{-2.5em}
		\includegraphics[width=10cm]{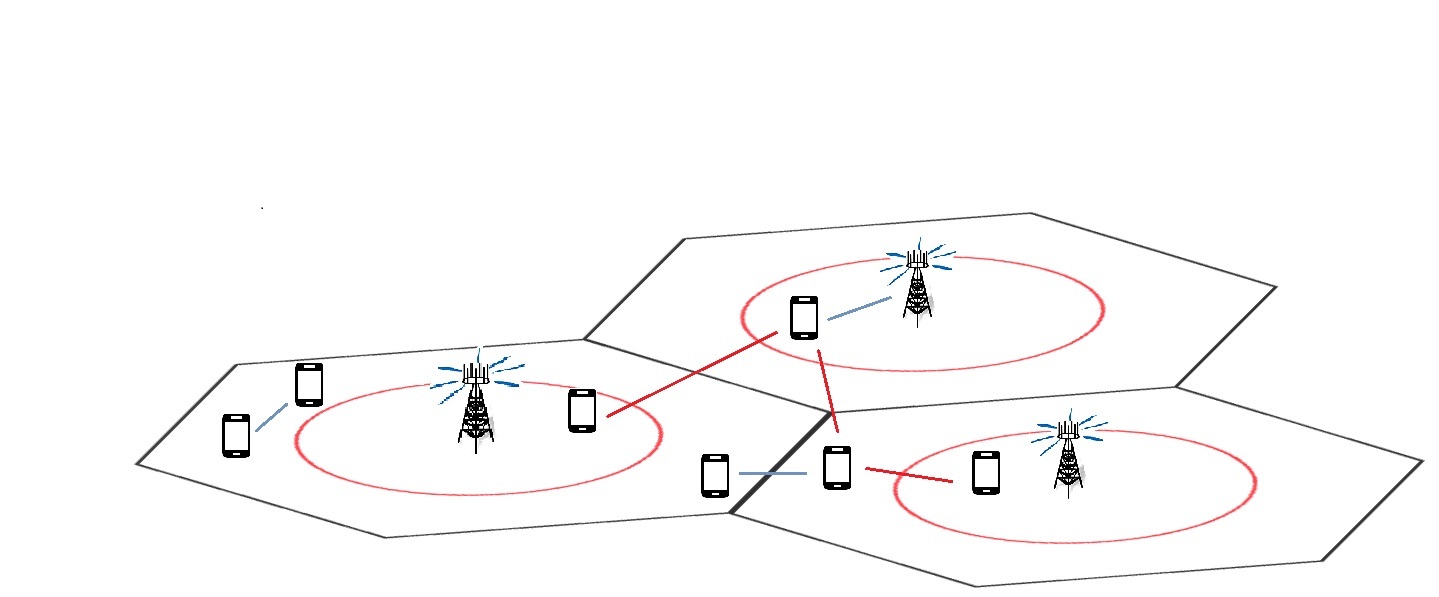}
	\label{Network Model}
		\caption{Network model}
\end{figure}

The cellular and D2D users locations are modeled as a Poisson Point Process (PPP) denoted $\Phi$ with density $\lambda$  . In our framework the users will select either cellular mode or D2D mode based on their distance from the serving Base station. The Base stations follow also a PPP  $\Phi_b$with density $\lambda_b$.
Each Base station will be serving users within a cell radius of $R_c$ and each Base station will be equipped with $M$ antennas.  
The D2D and cellular users transmit at  powers $P_d$ and $P_c$ respectively. We define $P_c =C P_d$ where C is a parameter calculated in order to optimize the system performance and $\lambda= a \times \lambda_b$.
Note that in this paper we always consider $ P_c > P_d $.
In fact,while D2D links are  more suited for short range communication the cellular links are utilized for longer range.
In our setting priority is given for the cellular users since they are also expected  to play the role of a relay point for D2D user wishing to communicate through the infrastructure of the cellular network, so the quality of the link of cellular user needs to be high. We can use, instead, some low power access points at the edge of the exclusion zone in order to relay data coming from D2D users to the cellular infrastructure.
This setting  enables us to guarantee the accessibility to the cellular infrastructure using low power links which should increase the energy efficiency of the network and reduce the interference. 
We will focus on the central Base station located at the origin.
The received signal at this base station during data transmission phase will be:

\begin{align} 
 \nonumber Y_c&& = \sum_{b}^{} \sum_{k}^{} \sqrt{P_c} \left\|r_{kb}\right\|^ {\frac{-\alpha}{2}} h_{kb} u_{kb}  +\\
\nonumber &&\sum_{i}^{}  \sqrt{P_d} \left\|r_{i}\right\|^ {\frac{-\alpha}{2}} h_{i} u_{i} + N_c  
\end{align} 

Where $\alpha >2$ denotes the pathloss exponent, $r_{kb}$ and $r_{i}$ denote, respectively, the distance of the kth cellular user in the bth cell and the ith D2D user from the reference Base station.
$h_{kb} \in C^ {M\times1}$ and $h_{i} \in C^ {M\times1}$ denote, respectively, the vector channels from the kth cellular user in the bth cell and the ith D2D user to the reference Base station with i.i.d $CN (0,1)$elements.
$u_{kb}$ and $u_{i}$ denote the zero mean unit variance transmitted data symbols of the kth cellular user in the bth cell and the ith D2D user.
$N_c$ is the additive white Gaussian noise at the reference Base station.

  The received signal at a reference D2D receiver in the reference Base station during the uplink phase is given by:
			\begin{align}
\nonumber {Y_{j}}^{d} &&= \sum_{b}^{} \sum_{k}^{} \sqrt{P_c} \left\|r_{jkb}\right\|^ {\frac{-\alpha}{2}} g_{jkb} u_{kb}  +\\
\nonumber &&\sum_{i}^{}  \sqrt{P_d} \left\|r_{ji}\right\|^ {\frac{-\alpha}{2}} g_{ji} u_{i} + N_d  
		\end{align}
		Where  $r_{jkb}$ and $r_{ji}$ denote, respectively, the distance of the kth cellular user in the bth cell and the ith D2D user from the jth D2D receiver taken here as a reference.
$g_{jkb} \in C^ {1\times1}$ and $g_{ji} \in C^ {1\times1}$ denote, respectively, the vector channels from the kth cellular user in the bth cell and the ith D2D user to the jth D2D receiver.
$N_d \in C^ {1\times1}$ is the additive white Gaussian noise at the jth D2D receiver.
 Gaussian signaling is assumed and all channel coefficients are i.i.d $CN(0,1)$ . 
In this paper we consider MRC  receivers which proved to be sufficient to achieve Massive MIMO gain.
We will focus on the imperfect CSI case where only cellular users channel is estimated.
 The idea developed here is to restraint the location of the cellular users within the cell  which will reduce on average the impact of pilot contamination. In fact, while in the classical network co-pilot users may be located anywhere within the cell,  in our framework these users will be within a circle of radius Re centered on the base station and thus increasing the average distance between co-pilot users.
We consider the Matrix $Q\in C^ {T_p \times T_p}$ containing the pilot sequences that satisfies $Q^*Q= I_{T_p}$ the same pilot sequences are reused in every cell with a reuse factor of 1. In this case the $M \times T_p$ dimensional received signal at the reference Base station during the training phase is:
\begin{align}
Y_{c}^p && = \nonumber \sum_{b}^{} \sum_{k}^{} \sqrt{P_c} \left\|r_{kb}\right\|^ {\frac{-\alpha}{2}} h_{kb} {q_{k}}^*  + \\
&&  \nonumber \quad\sum_{i}^{}  \sqrt{P_d} \left\|r_{i}\right\|^ {\frac{-\alpha}{2}} h_{i} u_{i}^{p*} + N_{c}^p  
\end{align}

with $u_{i}^{p}$ represent the  $T_p \times 1 $ data vector sent by the D2D transmitters during the training phase, $N_{c}^p$ the $M \times T_p$ dimensional additive Gaussian noise at the reference Base station and $q_{k} \in C^ {T_p\times1}$ the pilot vector used by user $k$ in cell $b$.

\section{Average cellular SINR \& average interference on D2D receivers}

Upon receiving the training signal, the Base station will estimate the channel of the kth cellular user in its cell as follow:
$$ \hat{h_{k1}}= Y_{c}^p q_{k}$$ 

The channel estimate in this case depends on whether D2D communication is allowed or not during training phase. In fact while allowing D2D transmitters to be active during this phase may allow us to use the system  resources efficiently it may on the other hand introduce a further  degradation of the channel estimate since, in addition to the usual limiting factor of pilot contamination, further degradation due to active D2D transmitters persist even in a Massive MIMO  system\cite{interplay}.\\
Two expressions of the estimate of the channel of  kth user can be derived:
With muted D2D the channel estimate will be:

 $$ \hat{h_{k1}}=  \sum_{b}^{}  \sqrt{P_c} \left\|r_{kb}\right\|^ {\frac{-\alpha}{2}} h_{kb} + N_{c}^p q_{k}  $$

If D2D communication was active during training phase the channel estimate is:

 $$ \hat{h_{k1}}=  \sum_{b}^{}  \sqrt{P_c} \left\|r_{kb}\right\|^ {\frac{-\alpha}{2}} h_{kb}+ \sum_{i}^{}  \sqrt{P_d} \left\|r_{i}\right\|^ {\frac{-\alpha}{2}} h_{i} u_{i}^{p*} q_{k} + N_{c}^p q_{k}  $$

\subsection{Interference analysis}

In the framework of our system  the reference base station will be subject to  interference coming both from D2D and cellular users. 
The cellular interference  will be generated by  co-pilot users in other cells. This is the usual limiting factor of Massive MIMO system the rest of the cellular interference vanish thanks to the large number of receive antennas at the Base station. Since the cellular interference results from the pilot contamination phenomenon, with a pilot reuse factor of $1$ the cellular interferer will have the same density $\lambda_b$ as the Base stations. \\

D2D interference will persist  if D2D transmission was active during the training phase, but the interfering signal will come from limited areas defined by the exclusion area radius. In fact in each cell the D2D terminals are allowed to be active only within a ring of inner radius Re and outer radius $R_c$. 
Limiting the area on which D2D interferer may span  changes the density of active D2D users since under our system assumptions the D2D users form now a point process referred to as  Poisson Hole Process which is a Cox process \cite{stocha}.\\
Definition 1 (Poisson hole process):\\
let $\phi_m$ and $\phi_n$ two PPPs with respective intensities  $\lambda_m$ and $\lambda_n$ with $\lambda_n >\lambda_m$.
For each $x \in \phi_m$, remove  all points  in $\phi_n \cap b(x,D)$ where $b(x,D)$ is a ball centered at $x$ and with radius $D$.
Then, the remaining points of $\phi_n$ form a Poisson hole process with density $\hat{\lambda_n}=\lambda_n \exp{(-\lambda_m \pi D^2)}$.

 The new density of the D2D users is then given by: 
$$\lambda_d = \lambda \exp{(-\pi \lambda_b Re^2)}$$

for a D2D receiver the interference will come from all cellular users and from the other D2D transmissions.We consider that any D2D receiver only knows the channel to its transmitter.

for a typical  D2D receiver during uplink  the cellular interference will come from cellular user spanning over circles with radius $R_e$  centered on the network Base stations. Then for the  D2D receiver the cellular interferer will have a density of:

$$\lambda_c = \lambda (1- \exp{(-\pi \lambda_b R_e^2)})$$ 
while the density of the D2D interferer is $\lambda_d$

\subsection{Impact of the exclusion zone on channel estimation}

In order to show the impact of the exclusion zone on the CSI acquisition we calculate the mean square error of channel estimation. We can show then, that  with the exclusion zone and deactivated D2D users during the training phase, we are able to reduce the Mean square error in channel estimation we also can see that this error increases when increasing the radius of the exclusion zone.
the mean square error of the channel estimation is given by:
$$ MSE = E(\epsilon_{channel} \epsilon_{channel} ^*) = E((\hat{h_{k1}}-h_{k1}) (\hat{h_{k1}}-h_{k1})^*) $$

with deactivated D2D transmission during training phase the $MSE$ after normalization is:
\begin{align}
\nonumber MSE & = E( (\sum_{b}^{}   \left\|r_{kb}\right\|^ {\frac{-\alpha}{2}} h_{kb} + N_{c}^p q_{k}-h_{k1})^2)\\ 
\nonumber& = M(\frac{1}{Pc}+ \frac{ 2 \pi \lambda_b (2R_c-R_e)^{2-\alpha}}{\alpha-2})	 
\end{align}

If the D2D transmission is active during the training phase the MSE will have an additional term coming from the D2D interference which will reduce the accuracy of the channel estimation. In this case the MSE becomes:

$$MSE = M(\frac{1}{Pc}+ \frac{ 2 \pi \lambda_b (2R_c-R_e)^{2-\alpha}}{\alpha-2}+\frac{  P_d 2 \pi \lambda_d (R_e)^{2-\alpha}}{(\alpha-2) P_c})	 $$

\subsection{Average cellular and D2D SINR }

In order  assess the performance gain generated by the exclusion region framework we will investigate the average SINR of a reference user located at the reference cell.
With the assumption that this user (k) is in cellular mode we have after receiving the uplink signal and using a MRC receiver\cite{interplay}:
\begin{align}
\nonumber&\lim_{M\to\infty}   \frac{1}{M} (Y_{c}^p q_{k})^* \times Y_{c} = \\
\nonumber&\sum_{b}^{}  P_c \left\|r_{kb}\right\|^ {{-\alpha}}  u_{kb}  + \sum_{i}^{}  P_d \left\|r_{i}\right\|^ {{-\alpha}} h_{i} (u_{i}^{p} q_{k})^* u_{i} 
\end{align}

Conditioned on the position of user $k$ taken here as a reference user we can have the following average SINR expression:

$$E (SINR_{k1}| r_{k1})  =\frac{{P_c}^2 \left\|r_{k1}\right\|^ {{-2 \alpha}}}{ \frac{{P_c}^2 2 \pi \lambda_b (2 R_c-R_e)^{2-2\alpha}}{2 \alpha-2}+\frac{{P_d}^2 2 \pi \lambda_d (R_e)^{2-2\alpha}}{2 \alpha-2} }$$

\begin{proof}
 See Appendix A.

\end{proof}
 We notice the persistence of the D2D interference even with an infinite number of antennas. This due to the fact that D2D interference is active during the training phase \cite{interplay}.
One way to remove the D2D interference is to mute D2D communication during training phase.In this case the average reference user SINR will become:

$$E (SINR_{k1}| r_{k1})  =\frac{{P_c}^2 \left\|r_{k1}\right\|^ {{-2 \alpha}}}{ \frac{{P_c}^2 2 \pi \lambda_b (2 R_c-R_e)^{2-2\alpha}}{2 \alpha-2} }$$

For computing the average SINR of a typical D2D receiver we consider a user in D2D mode  located at a distance $d$ from the reference base station and having a distance $r_{j}$ from its transmitter. 
for a typical D2D receiver the average SINR conditioned on its position is :
\begin{align}
\nonumber& E(SINR_{j}^{d}| r_{j},d) \\
\nonumber&= \frac{{P_d} \left\|d \right\|^ {{-\alpha}}}{ \sigma_{D2D}^2+\frac{{P_c}2 \pi \lambda_c (d-Re)^{2-\alpha}}{ \alpha-2}+\frac{{P_d} 2 \pi \lambda_d (r_0)^{2-\alpha}}{\alpha-2}}
\end{align}

\section{ Optimizing the cellular user average SINR}

Since the  major shortcoming of D2D underlaid cellular network is the cross tier interference. Optimizing the respective cellular and D2D transmission powers and the exclusion zone radius should enable us to extract the full potential of our system setting.
We propose an optimization problem in which transmission power and the radius of the exclusion zone are optimized.
We concentrated on the average interference and average SINR because we don't need perfect CSI in this case.
In fact while the usual Power Control problems require perfect knowledge of the user's channels, our problem only deals with the average interference which can be computed without exact knowledge of the channels. This is quite practical, since acquiring  channel information of all active users, both cellular and D2D, in the cell will cause a reduction in the data transmission phase.

Only basic parameters of the network are needed to find the optimal solution in our problem.
Taking $X= (R_e,C)$ as the vector of optimization variables, $f(X)= E(SINR_{ref}|r_{ref})$  and $g(X) =\frac{{P_c}2 \pi \lambda_c (d-Re)^{2-\alpha}}{ \alpha-2}+\frac{{P_d} 2 \pi \lambda_d (r_0)^{2-\alpha}}{\alpha-2}$
 we can formulate the optimization problem as:

\begin{equation*}
\begin{aligned}
& \underset{X}{\text{maximize}}
& & f(X) \\
& \text{subject to}
& & g(X) \leq I_{D2D}
\end{aligned}
\end{equation*}

Our optimization problem enable us to couple the exclusion region radius optimization with a power control problem allowing effective interference mitigation both for the D2D and cellular communication.
Although the optimization problem is not convex, the solution can be derived with zero duality gap.\\

\begin{lem}
The optimal solution of the proposed optimization problem is found using  KKT CONDITIONS. And it is given by the two following equations:
\begin{align}
\nonumber&\frac{4C \pi \lambda_d \left\|r_{k1}\right\|^ {{-2 \alpha}}(R_e)^{2-2\alpha}}{D^2(2 \alpha-2)} - \beta P_d \frac{2 \pi \lambda_c (d-Re)^{2-\alpha}}{ \alpha-2} =0\\
\nonumber&-\frac{C^2 2\pi}{2\alpha-2} \left\|r_{k1}\right\|^ {{-2 \alpha}}  [(2\alpha-2) \lambda_b C^2 (2R_c -R_e)^{1-2\alpha} + \lambda\\
\nonumber& \exp{(-\frac{R_e}{R_c})^2} ((2-2\alpha) R_e^{1-2\alpha}-\frac{2 R_e^{3-2\alpha}}{R_c^2})] - \beta D^2\\
\nonumber&[\frac{P_d 2 \pi }{\alpha-2} \times (C((\alpha-2) (d-R_e)^{1-\alpha}\lambda_c+ 2 \frac{R_e}{R_c^2}\\
\nonumber&(d-R_e)^{2-\alpha} \lambda \exp{(-\frac{R_e}{R_c})^2}) -\lambda_d r_0^{2-\alpha} 2 \frac{R_e}{R_c^2} )]=0
\end{align}
where $D= \frac{{C}^2 2 \pi \lambda_b (2 R_c-R_e)^{2-2\alpha}}{2 \alpha-2}+\frac{ 2 \pi \lambda_d (R_e)^{2-2\alpha}}{2 \alpha-2}$ and $\beta$ is the Lagrangian multiplier.
\end{lem}

\begin{proof}
 See Appendix B.

\end{proof}

\section{Numerical Results}
In this section we  present some numerical results to validate the derived analytical model  and to investigate the impact of the exclusion zone frame-work on the system performances depending on the system parameters.
 First we need to validate our analytical model. To this end we compare the simulation results for the average received interference at the Base station and mean square error in CSI acquisition with the results derived from the developed expressions of these two quantities. 
In our simulations we consider $31$ hexagonal cells with a side length $Rc=1Km$. The resulting Base station density can be approximated by $\lambda_b= \frac{1}{\pi R_c^2}$. for user density we consider $\lambda=a \lambda_b$. The exclusion region radius $R_e$ will vary from $\left[R_{emin},R_{eman}\right]$. We suppose that the reference user will be located at a distance $d=200m$ from its serving Base Station.\\

\begin{center}
\begin{tabular}{|c|l|c|l|}
\hline
$\alpha$ & $3$ & $ \lambda_b$ & $\frac{1}{\pi R_c^2}$\\
\hline
$R_c$& $1 Km$ & $\lambda$  & $150 \lambda_b$\\
\hline
$R_e$ & $ [0.4Km,0.9Km] $&  $\sigma^2 $&$ 10^-3$\\
\hline
\end{tabular}
\captionof{table}{Simulation Parameters}
\label{tab1}
\end{center}

In the simulation, we generate mobile users according to a PPP with density $\lambda$ These users will be randomly located in the network. Transmit power will be allocated to users according to there locations. The simulation scenario consider pilot aided channel estimation. The pilot reuse factor is 1. After channel estimation, data detection is performed.  
The received signal at the base station is averaged after 10000 iteration and measurements will be done for different values of $R_e$.  

\begin{figure}[h]
	\centering
		\includegraphics[width=7cm]{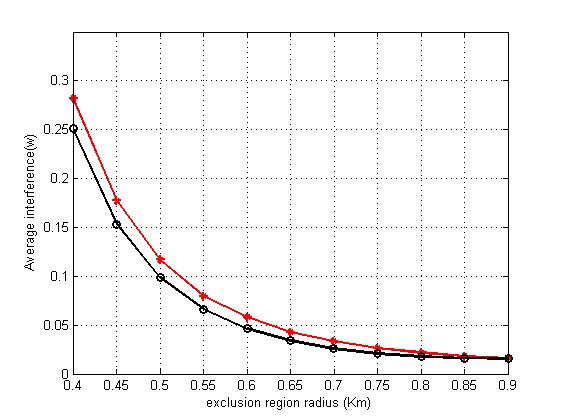}
	\label{Average interference}
		\caption{Average received interference at the reference BS }
\end{figure}

We can see that the analytical and simulation results match. The two curves represent respectively the analytical average received interference at the base station given by the denominator of the cellular user average SINR and the measured value after simulation.

For the MSE we compare the simulation curve with the analytical expression and with the MSE curve when no exclusion zone is implemented. When $R_e=R_c$ i.e the cellular user can be anywhere within the cell, the MSE curve represent the classical case without exclusion zone. We notice that using our system setting the MSE of channel estimation remains below the MSE of the classical case and they are only equal when  $R_e=R_c$.

\begin{figure}[h!]
	\centering
		\includegraphics[width=7cm]{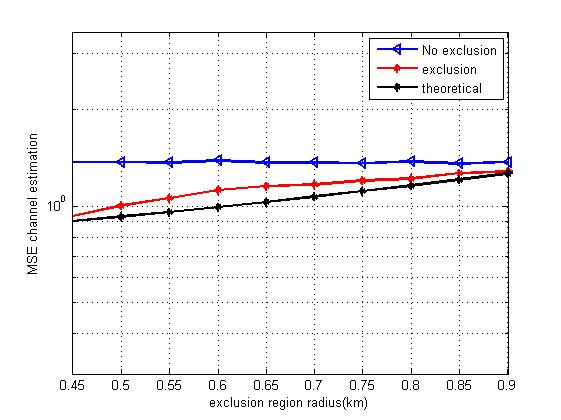}
	\label{Mean square Error CSI}
	\caption{Normalized Mean square Error in reference user channel estimation }
\end{figure}

The next figure shows the average SINR curve evolution as a function of$R_e$ for different user densities.
User density is an important parameter in the system. In fact while the cellular interferer always have as density $\lambda_b$ ,
the D2D interferer density increase while increasing $\lambda$

\begin{figure}[h]
	\centering
		\includegraphics[width=7cm]{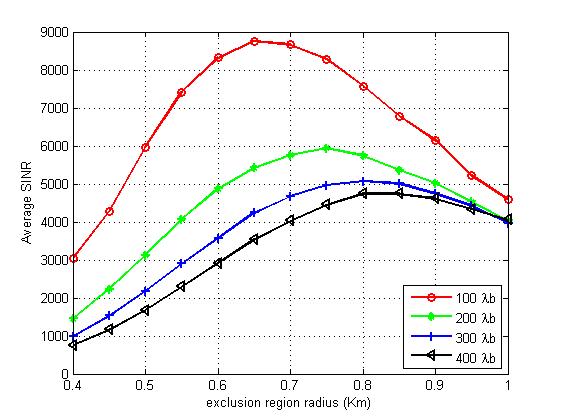}
	\label{fig:different densities}
	\caption{Average cellular user SINR for different users densities}
\end{figure}

We present, in table II, the optimal values  of the optimization variables for different interference thresholds on a reference D2D receiver located at $0.95 Km$ from the reference Base Station for a D2D transmission power of $16 dBm$ and a $C_{max} =10$ .
\begin{center}
\begin{tabular}{|c|l|c|l|}
\hline
$I_{D2D} (W)$& $12$ & $ 15$ & $18$\\
\hline
$R_e (Km)$& $0.7465 $ & $0.5994 $ & $0.6924 $\\
\hline
$C$ & $ 2.2176 $&  $8.6196 $&$ 10$\\
\hline
\end{tabular}
\captionof{table}{Solutions for different $I_{D2D}$}
\label{tab2}
\end{center}

\section{Conclusion}

In this paper, we investigated the interplay between  Massive MIMO and D2D underlaid cellular network where the mobile users choose their operating mode depending on their location. Based  on the exclusion zone framework  we investigated the average  user SINR  and the  mean square error in channel estimation.  We have shown that this system setting enables the reduction of the impact of pilot contamination. We also developed an analytical  model  that allow the optimization of the exclusion region radius and cellular user power in order to optimize the network performance. Although, the proposed framework allow the coexistence of cellular mode and D2D mode while minimizing the inter-tier interference, the fairness criterion should be investigated especially in terms of energy consumption.

\section{Appendix}

A. proof for average SINR

The average D2D interference at the reference Base Station is given by:
\begin{align}
 \nonumber&E\left(\left\|\sum_{i}^{}  P_d \left\|r_{i}\right\|^ {{-\alpha}} h_{i} (u_{i}^{p} q_{k})^* u_{i}\right\|^2 \right)= \\
\nonumber& E\left(\sum_{i}^{} P_d^2 \left\|r_{i}\right\|^ {{-2*\alpha}} \right)=\\
\nonumber& P_d^2 2 \pi \lambda_d  \int_{R_e}^{\infty} r^{1-2\alpha} \partial{r}=\\
\nonumber& \frac{{P_d}^2 2 \pi \lambda_d (R_e)^{2-2\alpha}}{2 \alpha-2}
\end{align}
The integral is due to Campbell's formula \cite{stocha}.
In a similar way we can derive the expression of the average cellular interference:
\begin{align}
 \nonumber&E\left(\left\|\sum_{b}^{} \sum_{k}^{} P_c \left\|r_{kb}\right\|^ {{-\alpha}}  u_{kb} \right\|^2 \right) = \\
\nonumber& E\left(\sum_{b}^{}  P_c^2 \left\|r_{kb}\right\|^ {{-2\alpha}}  \right)=\\
\nonumber& P_c^2 2 \pi \lambda_b  \int_{2R_c-R_e}^{\infty} r^{1-2\alpha} \partial{r}=\\
\nonumber& \frac{{P_c}^2 2 \pi \lambda_b (2 R_c-R_e)^{2-2\alpha}}{2 \alpha-2}
\end{align}

B. Proof of Lemma 1:

To prove Lemma 1 we use two propositions given by Arrow and Enthoven in \cite{Quasi} for quasi-concave optimization problems.
First we need to prove that the interference to which a typical D2D receiver is subject is a quasi convex function of the vector of the optimization variables and that the Objective function is also quasi concave as a function of the optimization variables.

In order to prove the strict quasi-concavity of the objective function it is sufficient to prove the strict concavity of  $ \Psi(X) = Log (f(X))$ \cite{boyd}
We compute the hessian $ H_{\Psi}$of $ \Psi(X) = Log (f(X))$ 
$$ H_{\Psi} =\nabla^2 \Psi(X) $$
taking $\lambda_b= \frac{1}{\pi R_c^2}$, $P_c= C \times P_d$ and $\lambda= a \lambda_b$ we have:
\begin{align} 
\nonumber \Psi(X) &=Log( 	C^2 R_c^2 (\alpha-1) \left\| r_{k1}\right\|^ {-2\alpha})-\\
\nonumber & Log (C^2 (2R_c -R_e)^{2-2\alpha}+a R_e^{2-2\alpha} \exp{(-\frac{R_e}{R_c})^2})
\end{align}

\begin{align}
\nonumber \frac{\partial^2 \Psi}{\partial^2 C}&= \frac{-2}{C^2} - \frac{2 (2R_c -R_e)^{2-2\alpha}}{C^2 (2R_c -R_e)^{2-2\alpha}+a  \nonumber R_e^{2-2\alpha} \exp{(-\frac{R_e}{R_c})^2}} +\\
\nonumber & \frac{4 C^2 (2R_c -R_e)^{4-4\alpha}}{(C^2 (2R_c -R_e)^{2-2\alpha}+a R_e^{2-2\alpha} \exp{(-\frac{R_e}{R_c})^2})^2}
\end{align}

$$\frac{\partial^2 \Psi}{\partial^2 C} < 0$$

$$\frac{\partial^2 \Psi}{\partial^2 R_e}= \frac {-A \times D + E^2} {D^2}$$
with 
\begin{align}
\nonumber A&=(2\alpha-2)(2\alpha-1)  C^2 (2R_c -R_e)^{-2\alpha} + a \exp{(-\frac{R_e}{R_c})^2} \times\\
\nonumber & ( (1-2\alpha)(2-2\alpha) R_e^{-2\alpha} -(2-2\alpha) 2 \frac{R_e^{2-2\alpha} }{R_c^2}-\\
\nonumber &(3-2\alpha) 2 \frac{R_e^{2-2\alpha} }{R_c^2}-\frac{4 R_e^{4-2\alpha} }{R_c^4}) \\
\nonumber E&=(2\alpha-2)  C^2 (2R_c -R_e)^{1-2\alpha} + \\
\nonumber &a \exp{(-\frac{R_e}{R_c})^2} ((2-2\alpha) R_e^{1-2\alpha}-\frac{2 R_e^{3-2\alpha}}{R_c^2})\\
\nonumber D&=C^2 (2R_c -R_e)^{2-2\alpha}+a R_e^{2-2\alpha} \exp{(-\frac{R_e}{R_c})^2}
\end{align}

$$\frac{\partial^2 \Psi}{\partial^2 R_e} < 0$$

\begin{align}
\nonumber & \frac{\partial^2 \Psi}{\partial C \partial R_e} =\\
\nonumber & \frac{-D\times (2\alpha-2) 2 C (2R_c -R_e)^{1-2\alpha} + E \times 2 C (2R_c -R_e)^{2-2\alpha} }{D^2}
\end{align}
$$\frac{\partial^2 \Psi}{\partial C \partial R_e} < 0$$

\begin{align}
\nonumber & \frac{\partial^2 \Psi}{ \partial R_e \partial C} =\\
\nonumber &\frac{- D \times ((2\alpha-2)  C 2 (2R_c -R_e)^{1-2\alpha})+E \times(2 C (2R_c -R_e)^{2-2\alpha})}{D^2}
\end{align}

$$\frac{\partial^2 \Psi}{ \partial R_e \partial C} < 0$$

In addition we have  $\frac{\partial^2 \Psi}{ \partial R_e \partial C} > \frac{\partial^2 \Psi}{\partial^2 R_e} $ and 
 $\frac{\partial^2 \Psi}{ \partial R_e \partial C} > \frac{\partial^2 \Psi}{\partial^2 C} $
 Then $ H_{\Psi}$ is a negative semi-definite matrix which proves the strict quasi-concavity of $f(X)$

For the constraint $g(X)$ it is sufficient to show that the interference is an increasing function of the exclusion radius $R_e$
in fact

\begin{align}
\nonumber\frac{\partial g}{\partial R_e} &=\\
\nonumber & \frac{P_d 2 \pi \lambda}{\alpha-2} \times (C((\alpha-2) (d-R_e)^{1-\alpha} (1-\exp{(-\frac{R_e}{R_c})^2}) +\\
\nonumber & 2 \frac{R_e}{R_c^2}(d-R_e)^{2-\alpha} \exp{(-\frac{R_e}{R_c})^2}) - r_0^{2-\alpha} 2 \frac{R_e}{R_c^2} )
\end{align}
we can see that this derivative is positive in the interesting regime where $ P_c > P_d$\\
$\frac{\partial g}{\partial R_e} > 0$ iff $C > C_{min}$
$$\frac{\partial g}{\partial C} > 0$$
 
Since $\nabla g (X) > 0 \;\;$, $g (X)$ is a quasi-convex function \\
\\

Now, we use the following propositions:\\

Proposition 1 (Arrow and Enthoven):
Assume that $g_1, . . . , g_m$ are quasi-concave functions and that the following

regularity conditions hold:
\begin{enumerate}
\item Slater condition is verified : $ \exists X_0 \in C^2  \; such that \; g_i(X_0) > 0 \; \forall i$
\item for each i  $g_i$ is either concave or $\nabla g_i (x) \neq 0 $for each feasible solution of the optimization problem 
\end{enumerate}
Then if $ X^* $ is a local optimal solution of the optimization problem then $\exists \gamma^* $ such that with $(X^*,\gamma^*)$ 
the Kuhn-Tucker conditions hold.

Proposition 2 (Arrow and Enthoven):
If in addition the objective function f is twice differentiable on the feasible set and additionally $\nabla f(X^*) \neq 0$ holds then  $X^*$ is an optimal solution of the optimization problem.\\

It is clear that our optimization problem satisfies the optimality conditions stated in the two propositions which means that the optimal exclusion zone radius and cellular power can be found using KKT conditions.



\end{document}